\documentclass{theoretics}
\usepackage{ifxetex}
\usepackage{bm}
\usepackage{array}
\usepackage{cleveref}
\usepackage{makecell}
\usepackage{float}
\renewcommand{\epsilon}{\varepsilon}

\def\Pr{\mathop{\mathbf{Pr}}\nolimits}
\def\Ex{\mathop{\mathbf{{}E}}\nolimits}
\def\Var{\mathop{\mathbf{Var}}\nolimits}

\newcommand{\abs}[1]{\left\vert#1\right\vert}

\newcommand{\tuple}[1]{\left(#1\right)}

\newcommand{\tp}{\tuple}

\newcommand{\numP}{\ensuremath{\#\mathbf{P}}}

\newcommand{\defeq}{:=}

\newcommand{\DTV}[2]{d_{\mathrm{TV}}({#1},{#2})}

\def\*#1{\bm{#1}} % Use \*A for \mathbf{A}
\def\+#1{\mathcal{#1}} % Use \+A for \mathcal{A}
\def\-#1{\mathrm{#1}} % Use \-A for \mathrm{A}
\def\=#1{\mathbb{#1}} % Use \=A for \mathbb{A}

\usepackage[textsize=tiny]{todonotes}

\usepackage{xifthen}

\makeatletter
\def\prob#1#2#3{\goodbreak\begin{list}{}{\labelwidth\z@ \itemindent-\leftmargin
                        \itemsep\z@  \topsep6\p@\@plus6\p@
                        \let\makelabel\descriptionlabel}
                \item[\it Name]#1
               \item[\it Instance]                #2
                \item[\it Output]#3
                \end{list}}
\makeatother

\title{A simple polynomial-time approximation algorithm for the total variation distance between two product distributions}

\ThCSthanks{Mark Jerrum was supported by grant EP/S016694/1 `Sampling in hereditary classes' from the Engineering and Physical Sciences Research Council (EPSRC) of the UK.
Weiming Feng, Heng Guo, and Jiaheng Wang have received funding from the European Research Council (ERC) under the European Union's Horizon 2020 research and innovation programme (grant agreement No.~947778).
 Jiaheng Wang has also received financial support from an Informatics Global PhD
Scholarship at The University of Edinburgh.}

\ThCSauthor[edinburgh]{Weiming Feng}{wfeng@ed.ac.uk}[https://orcid.org/0000-0003-4636-1023]
\ThCSauthor[edinburgh]{Heng Guo}{hguo@inf.ed.ac.uk}[https://orcid.org/0000-0001-8199-5596]
\ThCSauthor[london]{Mark Jerrum}{m.jerrum@qmul.ac.uk}[https://orcid.org/0000-0003-0863-7279]
\ThCSauthor[edinburgh]{Jiaheng Wang}{jiaheng.wang@ed.ac.uk}[https://orcid.org/0000-0002-5191-545X]
 
\ThCSaffil[edinburgh]{School of Informatics, University of Edinburgh,
  United Kingdom}
\ThCSaffil[london]{School of Mathematical Sciences, Queen Mary,
  University of London, United Kingdom}

\ThCSyear{2023}
\ThCSarticlenum{8}
\ThCSdoi{10.46298/theoretics.23.8}
\ThCSreceived{Dec 22, 2022}
\ThCSaccepted{May 15, 2023}
\ThCSpublished{June 12, 2023}
\ThCSkeywords{total variation distance, product distribution, approximation algorithm}
\ThCSshortnames{W.\ Feng, H.\ Guo, M.\ Jerrum, J.\ Wang} 
\ThCSshorttitle{Approximating the total variation distance between two product distributions}

\begin{document}

\maketitle

\begin{abstract}
  We give a simple polynomial-time approximation algorithm for the total variation distance between two product distributions.
\end{abstract}

\section{Introduction}
The total variation (TV) distance is a fundamental metric to measure the difference between two distributions.
It is essentially the $L^1$ distance.
Unlike many other quantities for similar uses, such as the relative entropy and the $\chi^2$-divergence,
the TV distance does not tensorise over product distributions.
In fact, it was discovered recently that, somewhat surprisingly, exact computation of the total variation distance, even between product distributions over the Boolean domain, is \numP-hard~\cite{BGMMPV22}.

This leaves open the question of approximation complexity of the TV distance.
In \cite{BGMMPV22}, the authors give polynomial-time randomised approximation algorithms in two special cases over the Boolean domain,
when one of the distribution has marginals over $1/2$ and dominates the other,
or when one of the distribution has a constant number of distinct marginals.
Their method is based on Dyer's dynamic programming algorithm for approximating the number of knapsack solutions \cite{Dye03}. 

In this note, we give a simple polynomial-time approximation algorithm for total variation distance between two product distributions.
Our algorithm is based on the Monte Carlo method and does not have further restrictions.

\begin{theorem}\label{theorem-main} 
  Let $[q] = \{1,2,\ldots,q\}$ be a finite set. 
  There exists an algorithm such that given two product distributions $P,Q$ over $[q]^n$ and parameters $\epsilon > 0$ and $0 <\delta < 1$, 
  it outputs a random value $\widehat{d}$ in time $O(\frac{n^2}{\epsilon^2} \log \frac{1}{\delta})$ such that $(1-\epsilon)\DTV{P}{Q} \leq \widehat{d} \leq (1+\epsilon)\DTV{P}{Q} $ holds with probability at least $1-\delta$.
\end{theorem}

Our algorithm can also handle the case where each coordinate has a different domain size without any change.
In \Cref{theorem-main}, the input product distributions are given by the marginal probability for each coordinate and each $c\in[q]$ in binary.
The stated running time assumes that all arithmetic operations can be done in $O(1)$ time.

To approximate the TV distance, the na\"{i}ve Monte Carlo algorithm works well when the two distributions are sufficiently far away.
However, when the TV distance is exponentially small, na\"{i}ve Monte Carlo may require exponentially many samples to return an accurate estimate.
Our idea is to consider a distribution that can be efficiently sampled from and yet boosts the probability that the two distributions are different.
Ideally, we would want to use the optimal coupling, but that is difficult to compute.
We use instead the coordinate-wise greedy coupling as a proxy,
where each coordinate is coupled optimally independently.
We further condition on the (potentially very unlikely) event that the two samples are different.
Normally, conditioning on an unlikely event is a bad move since computational tasks would become hard.
However, here they are still easy thanks to the independence of the coordinates under the coupling.
With this conditional distribution, our estimator is the ratio between the probabilities of the assignment in the optimal coupling and in the greedy coupling.
We show that this estimator is always bounded from above by $1$ and its expectation is at least $1/n$.
This means that the standard Monte Carlo method will succeed with high probability using only polynomially many samples.

One remaining question is if a deterministic approximation algorithm exists for the TV distance. 
The answer might be positive, because of the connection with counting knapsack solutions established by Bhattacharyya, Gayen, Meel, Myrisiotis, Pavan, and Vinodchandran \cite{BGMMPV22},
and the deterministic approximation algorithm for the latter problem
by Gopalan, Klivans, Meka, \v{S}tefankovi\v{c}, Vempala, and Vigoda~\cite{GKM10,GKMSVV11,SVV12}.

\section{Preliminaries}

Let $\Omega$ be a (finite) state space, and $P$ and $Q$ be two distributions over $\Omega$.
The total variation distance is defined by 
\begin{align*}
\DTV{P}{Q}\defeq\frac{1}{2}\sum_{\omega\in\Omega}\abs{P(\omega)-Q(\omega)}.
\end{align*}
It satisfies the following:
\begin{itemize}
  \item for any event $A\subseteq \Omega$, $\DTV{P}{Q}\ge \abs{P(A)-Q(A)}$;
  \item for any coupling $\+C$ between $P$ and $Q$, $\DTV{P}{Q}\le \Pr_{\+C}[X\neq Y]$,
    where $X \sim P$ and $Y\sim Q$.
\end{itemize}
In particular, there exists an event $A_O$ and an optimal coupling $\+O$ such that $\DTV{P}{Q}=\abs{P(A_O)-Q(A_O)}=\Pr_{\+O}[X\neq Y]$.
Optimal couplings are not necessarily unique. For any optimal coupling $\+O$, it holds that
\begin{align}\label{eqn:opt-same}
\forall \omega\in\Omega,\quad   \Pr_{\+O}[X=Y=\omega] = \min\{P(\omega),Q(\omega)\}.
\end{align}
The above equation holds because (1) for any valid coupling $\+C$, it holds that $\Pr_{\+C}[X=Y=\omega] \leq  \min\{P(\omega),Q(\omega)\}$; (2) to achieve the optimal coupling, every $\omega$ must achieve the equality.
We have
\begin{align}  \label{eqn:opt-diff}
  \Pr_{\+O}[X=\omega\wedge Y\neq X] = \Pr_{\+O}[X=\omega]-\Pr_{\+O}[X=Y=\omega] = \max\{0,P(\omega)-Q(\omega)\}.
\end{align}

\section{Algorithm}

From now on we consider only product distributions.
Let $\Omega=[q]^n$ be the state space, where $[q]=\{1,\dots,q\}$ is a finite set.
Let $P=P_1\otimes P_2\otimes \dots \otimes P_n$ and $Q=Q_1\otimes Q_2\otimes \dots \otimes Q_n$ be two product distributions.
%Assume that $\DTV{P_i}{Q_i}>0$ for all $1\le i\le n$,
%as otherwise we may remove such $i$ without changing the total variation distance.
Let $\+O$ be an (arbitrary) optimal coupling between $P$ and $Q$.

Let $\+C$ be the coordinate-wise greedy coupling.
Namely, for each coordinate $i$ and $c\in [q]$, $\Pr_{\+C}[X_i=Y_i=c]=\min\{P_i(c),Q_i(c)\}$,
and the remaining probability can be assigned arbitrarily as long as $\+C$ is a valid coupling (but each coordinate is independent).
In other words, for each $i \in [n]$, $\+C$ couples $P_i$ and $Q_i$ optimally and independently.
Note that
\begin{align}\label{eqn:Pr-C-neq}
  \Pr_{\+C}[X\neq Y] = 1-\Pr_{\+C}[X = Y] = 1-\prod_{i=1}^n ( 1 - \DTV{P_i}{Q_i})
\end{align}
can be computed exactly.

Consider the distribution $\pi$ such that
\begin{align}  \label{eqn:pi-def}
  \pi(\omega)\defeq\Pr_{\+C}[X=\omega\mid X\neq Y].
\end{align}
We may assume $P$ and $Q$ are not identical, as otherwise the algorithm just outputs $0$. 
This makes sure that the distribution $\pi$ is well-defined. 
The following lemma shows that we can draw random samples from $\pi$ efficiently.
\begin{lemma}  \label{lem:sampling}
  We can sample from the distribution $\pi$ in $O(n)$ time.
\end{lemma}
\begin{proof}
We draw a random sample $\omega \in [q]^n$ from $\pi$ index by index. In the $k$-th step, where $1 \leq k \leq n$, we sample $\omega_k \in [q]$ from $\pi_k(\cdot\mid \omega_1,\omega_2,\ldots,\omega_{k-1})$, which is the marginal distribution on the $k$-th variable conditional on the values of the first $k-1$ variables being $\omega_1,\omega_2,\ldots,\omega_{k-1}$. By definition,
\begin{align*}
\pi_k(\omega_k\mid \omega_1,\omega_2,\ldots,\omega_{k-1}) = 	\frac{\Pr_{X\sim \pi}[\forall 1 \leq i \leq k, X_i = \omega_i]}{\Pr_{X\sim \pi}[\forall 1 \leq i \leq k-1, X_i = \omega_i]}.
\end{align*}
As $\omega_1,\ldots,\omega_{k-1}$ are sampled from the marginal distribution of $\pi$, the denominator is positive.
We show how to compute the numerator next, and the denominator can be computed similarly. By definition 
\begin{align*}
	& \Pr_{X\sim \pi}[\forall 1 \leq i \leq k, X_i = \omega_i] = \Pr_{(X,Y)\sim \+C}[\forall 1 \leq i \leq k, X_i = \omega_i \mid X \neq Y ]\\
\text{(by Bayes' law)}\quad	=\,& \tp{1-\Pr_{(X,Y)\sim \+C}[X = Y \mid \forall 1 \leq i \leq k, X_i = \omega_i ]} \cdot \frac{\prod_{i = 1}^k P_i(\omega_i)}{1-\prod_{i=1}^n (1-\DTV{P_i}{Q_i}) }.
\end{align*}
In the coupling $\+C$, every pair of $X_i$ and $Y_i$ is coupled optimally and independently. We have 
\begin{align}\label{eq-prod-coupling}
\Pr_{(X,Y)\sim \+C}[X = Y \mid \forall 1 \leq i \leq k, X_i = \omega_i] &=  \prod_{i=1}^k\frac{\Pr_{\+C}[X_i = Y_i = \omega_i]}{\Pr_{\+C}[X_i = \omega_i]} 	\prod_{i = k+1}^n\Pr_{\+C}[X_i = Y_i] \notag\\
 (\text{by~\eqref{eqn:opt-same}})\quad &= \prod_{i=1}^k\frac{\min\{P_i(\omega_i),Q_i(\omega_i)\}}{P_i(\omega_i)} 	\prod_{i = k+1}^n(1 - \DTV{P_i}{Q_i}).
\end{align} 
Combining the two equations, we can compute $\Pr_{X\sim \pi}[\forall 1 \leq i \leq k, X_i = \omega_i]$, and thus we can compute and sample from  $\pi_k(\cdot\mid \omega_1,\omega_2,\ldots,\omega_{k-1})$.
When sampling from the distribution $\pi$, we pre-process $\prod_{i = k+1}^n(1 - \DTV{P_i}{Q_i})$ for all $k$, and maintain the prefix products $\prod_{i=1}^k \min\{P_i(\omega_i),Q_i(\omega_i)\}$ and $\prod_{i=1}^k {P_i(\omega_i)}$.
This way, each conditional marginal distribution can be computed with $O_{q}(1)$ incremental cost. 
Hence, the total running time is $O_q(n)$, where $O_{q}(\cdot)$ hides a factor linear in $q$. 
\end{proof}

Let $\omega$ be a random sample from $\pi$.
Now consider the following estimator:
\begin{align}  \label{eqn:est-def}
  f(\omega)\defeq \frac{\Pr_{\+O}[X=\omega\wedge X\neq Y]}{\Pr_{\+C}[X=\omega\wedge X\neq Y]} =  \frac{ \max\{0,P(\omega)-Q(\omega)\}}{\Pr_{\+C}[X=\omega\wedge X\neq Y]},
\end{align}
where the second equality is due to \eqref{eqn:opt-diff}.
This estimator $f$ is well-defined,
because when $\Pr_{\+C}[X=\omega\wedge X\neq Y]=0$, $\pi(\omega)=0$ as well and $\omega$ will not be drawn.

In fact, if $\pi(\omega) = 0$, or equivalently $\Pr_{\+C}[X=\omega\wedge X\neq Y]=0$, 
it must be that $\max\{0,P(\omega)-Q(\omega)\}=0$.
This is because $\Pr_{\+C}[X=\omega\wedge X\neq Y]=0$ implies that either $\Pr_{\+C}[X=\omega]=P(\omega)=0$ or $\Pr_{\+C}[X\neq Y\mid X=\omega]=0$.
In the first case, $\max\{0,P(\omega)-Q(\omega)\}=0$.
In the second case $\Pr_{\+C}[Y=\omega\mid X=\omega]=1$,
which implies that $Q(\omega)\ge P(\omega)$, and $\max\{0,P(\omega)-Q(\omega)\}=0$ as well.

\begin{lemma}  \label{lem:est-compt}
  For any $\omega\in\Omega$ with $\pi(\omega) > 0$, $f(\omega)$ can be computed in $O(n)$ time.
\end{lemma}

\begin{proof}
  Note that
  \[
    \Pr_{\+C}[X=\omega\land X\neq Y]=P(\omega)\Pr_{\+C}[X\neq Y\mid X=\omega]=P(\omega)(1-\Pr_{\+C}[X=Y\mid X=\omega]). 
  \]
  Since $\pi(\omega) > 0$, it holds that $P(\omega) > 0$. Using~\eqref{eq-prod-coupling}, we have 
  \[
    f(\omega)=\max\left\{0,\frac{1-\frac{Q(\omega)}{P(\omega)}}{\frac{1}{P(\omega)}\Pr_{\+C}[X=\omega\land X\neq Y]}\right\}=\max\left\{0,\frac{1-\prod_{i=1}^{n}\frac{Q_i(\omega_i)}{P_i(\omega_i)}}{1-\prod_{i=1}^{n}\frac{\min\{P_i(\omega_i),Q_i(\omega_i)\}}{P_i(\omega)}}\right\}, 
  \]
  which can be computed in $O(n)$ time. 
\end{proof}

\begin{lemma}  \label{lem:est-prop}
  We have the following:
  \begin{align}
    \Ex_{\pi}f &= \frac{\Pr_{\+O}[X\neq Y]}{\Pr_{\+C}[X\neq Y]};  \label{eqn:est-bias}\\
    \frac{1}{n}&\le \Ex_{\pi}f\le 1. \label{eqn:est-exp-bound}
  \end{align}
  Moreover, for any $\omega\in\Omega$ with $\pi(\omega) > 0$,
  \begin{align}    \label{eqn:est-bound}
    0\le f(\omega)\le 1,
  \end{align}
  and it holds that
  \begin{align}\label{eqn:var<ex}
  	\Var_{\pi}f \leq \Ex_{\pi}f.
  \end{align}
\end{lemma}
\begin{proof}
  For \eqref{eqn:est-bias},
  Let $\Omega_+=\{\omega\in\Omega\mid\pi(\omega)>0\}$.
  Then,
  \begin{align*}
    \Ex_{\pi}f &= \sum_{\omega\in\Omega_+} \pi(\omega) \times \frac{\Pr_{\+O}[X=\omega\wedge X\neq Y]}{\Pr_{\+C}[X=\omega\wedge X\neq Y]}\\
    & = \sum_{\omega\in\Omega_+} \frac{\Pr_{\+C}[X=\omega\wedge X\neq Y]}{\Pr_{\+C}[X\neq Y]} \times \frac{\Pr_{\+O}[X=\omega\wedge X\neq Y]}{\Pr_{\+C}[X=\omega\wedge X\neq Y]}\\
    & =  \frac{\sum_{\omega\in\Omega_+}\Pr_{\+O}[X=\omega\wedge X\neq Y]}{\Pr_{\+C}[X\neq Y]} = \frac{\Pr_{\+O}[X\neq Y]}{\Pr_{\+C}[X\neq Y]},
  \end{align*}
  where in the last equation we used the aforementioned fact that $\pi(\omega)=0$ implies $\max\{0,P(\omega)-Q(\omega)\}=0$.

  For \eqref{eqn:est-exp-bound}, as $\+O$ is the optimal coupling, $\Pr_{\+O}[X\neq Y]\le\Pr_{\+C}[X\neq Y]$.
  For the other direction, notice that $\+O$ projected to coordinate $i$, denoted $\+O_i$, is a coupling between $P_i$ and $Q_i$.
  Thus,
  \begin{align*}
    \Pr_{\+O}[X\neq Y] \ge \max_{1\le i\le n} \Pr_{\+O_i}[X_i\neq Y_i] \ge \max_{1\le i\le n}\DTV{P_i}{Q_i}\,.
  \end{align*}
  On the other hand, by the union bound,
  \begin{align*}
    \Pr_{\+C}[X\neq Y] \le \sum_{i=1}^n \Pr_{\+C_i}[X_i\neq Y_i] = \sum_{i=1}^n\DTV{P_i}{Q_i} \le n \max_{1\le i\le n}\DTV{P_i}{Q_i}.
  \end{align*}

  For \eqref{eqn:est-bound}, the lower bound is trivial.
  For the upper bound, we only need to consider $\omega \in \Omega_+$ such that $P(\omega)>Q(\omega)$.
  In this case
  \begin{align*}
    f(\omega) & = \frac{ \max\{0,P(\omega)-Q(\omega)\}}{\Pr_{\+C}[X=\omega\wedge X\neq Y]} = \frac{P(\omega)-Q(\omega)}{\Pr_{\+C}[X=\omega]\Pr_{\+C}[X\neq Y\mid X=\omega]} \\
    &= \frac{P(\omega)-Q(\omega)}{P(\omega)(1-\Pr_{\+C}[X=Y\mid X=\omega])} = \frac{1-\frac{Q(\omega)}{P(\omega)}}{1-\Pr_{\+C}[X=Y\mid X=\omega]}.
  \end{align*}
  Since $\+C$ couples each coordinate independently,
  \begin{align*}
    \Pr_{\+C}[X=Y\mid X=\omega] = \prod_{i=1}^n \frac{\min\{P_i(\omega_i),Q_i(\omega_i)\}}{P_i(\omega_i)} \le \prod_{i=1}^n \frac{Q_i(\omega_i)}{P_i(\omega_i)}=\frac{Q(\omega)}{P(\omega)}.
  \end{align*}
  This finishes the proof of \eqref{eqn:est-bound}.
  
  For~\eqref{eqn:var<ex}, since $0 \leq f(\omega) \leq 1$ for all $\Omega \in \Omega_+$, $f(\omega)^2 \leq f(\omega)$ and thus $\Ex_{\pi}f^2 \leq \Ex_{\pi}f$. We have
  \begin{align*}
  \Var_{\pi}f = \Ex_{\pi}f^2 - (\Ex_\pi f)^2 \leq 	\Ex_{\pi}f^2 \leq \Ex_{\pi}f. &\qedhere
  \end{align*}
\end{proof}

\Cref{lem:est-prop} implies that standard Monte Carlo method can be used to accurately estimate $\Ex_{\pi}f = \frac{\Pr_{\+O}[X\neq Y]}{\Pr_{\+C}[X\neq Y]}$.
To implement the Monte Carlo algorithm, we use \Cref{lem:sampling} and \Cref{lem:est-compt}.

To be more specific, our approximate algorithm is to compute the median of means. 
The input contains the descriptions of $2n$ distributions $P_1,P_2,\ldots,P_n, Q_1,Q_2,\ldots,Q_n$ together with two parameters $\epsilon > 0$ and $0 < \delta < 1$.  
The algorithm proceeds as follows: 
\begin{itemize}
	\item for each $i$ from $1$ to $m = \lceil\frac{10n}{\epsilon^2}\rceil$, independently sample $\omega_i \sim \pi$ and let
	\begin{align*}
		F = \frac{1}{m}\sum_{i=1}^m f(\omega_i);
	\end{align*}
	\item use independent samples to compute $F$ for $s = 10\lceil \log\frac{1}{\delta} \rceil $ times to get $F_1,F_2,\ldots,F_s$ and let 
	\begin{align*}
		\widehat{F} = \text{Median}\{F_1,F_2,\ldots,F_s\};
	\end{align*} 
	\item output the value $\widehat{d} = (1-\prod_{i=1}^n ( 1 - \DTV{P_i}{Q_i}))\widehat{F}  $.
\end{itemize}

We claim that 
\begin{align}\label{eq-F}
	\Pr\left[ \abs{F - \Ex_{\pi}f} \geq \epsilon \Ex_{\pi}f \right] \leq \frac{1}{10}.
\end{align}
Assuming that \eqref{eq-F} holds, by the Chernoff bound, it holds that
\begin{align*}
	\Pr\left[ \abs{\widehat{F} - \Ex_{\pi}f} \geq \epsilon \Ex_{\pi}f \right] \leq \delta.
\end{align*}
Using~\eqref{eqn:est-bias} in  \Cref{lem:est-prop} and \eqref{eqn:Pr-C-neq}, we have
\begin{align*}
\Pr\left[ \abs{\widehat{d}-\DTV{P}{Q}}\ge \epsilon \DTV{P}{Q}\right] = \Pr\left[ \abs{\widehat{F} - \Ex_{\pi}f} \geq \epsilon \Ex_{\pi}f \right]\leq \delta.
\end{align*}
By~\Cref{lem:sampling} and \Cref{lem:est-compt}, the total running time is $O(nms) = O(\frac{n^2}{\epsilon^2}\log\frac{1}{\delta})$. This proves \Cref{theorem-main}.

Finally, we prove the claim~\eqref{eq-F}. Note that the expectation and the variance of the random variable $F$ satisfy that $\Ex F = \Ex_{\pi}f$ and $\Var{F} = \frac{1}{m}\Var_{\pi}f$. By Chebyshev's inequality, 
\begin{align*}
\Pr\left[ \abs{F - \Ex_{\pi}f} \geq \epsilon \Ex_{\pi}f \right] &= 	\Pr\left[ \abs{F - \Ex F} \geq \epsilon \Ex F \right] \leq \frac{\Var F}{\epsilon^2 (\Ex F)^2} = \frac{\Var_{\pi}f }{m\epsilon^2 (\Ex_{\pi}f)^2 }\\
\tag{by~\eqref{eqn:var<ex},~\eqref{eqn:est-exp-bound}, and~$m = \lceil\frac{10n}{\epsilon^2}\rceil$} &\leq  \frac{1}{m\epsilon^2 \Ex_{\pi}f}\leq \frac{n}{m\epsilon^2}\leq \frac{1}{10}.\\
%\tag{by~\eqref{eqn:est-exp-bound}} &\leq \frac{n}{m\epsilon^2}\\
%\tag{by~$m = \lceil\frac{10n}{\epsilon^2}\rceil$ } &\leq \frac{1}{10}.
\end{align*}

\printbibliography

\end{document}